\documentclass{article}

\usepackage{natbib}

\usepackage{amsmath, amssymb, amsfonts}
\usepackage{graphicx}
\usepackage{algorithm}
\usepackage{algorithmic}
\usepackage{multirow}
\usepackage{todonotes}
\usepackage{color}
\usepackage{hyperref}

\newtheorem{theorem}{Theorem}[section]
\newtheorem{definition}{Definition}[section]

\newtheorem{lemma}{Lemma}[section]
\newtheorem{proof}{Proof}[section]
\newcommand{\ab}[2]{\ensuremath{\langle #1,#2 \rangle}}

\begin{document}

\title{Factorial graphical lasso for dynamic networks}
\author{Ernst Wit, Antonino Abbruzzo\\University of Groningen}
\maketitle



\begin{abstract}
Dynamic networks models describe a growing number of important scientific processes, from cell biology and epidemiology to sociology and finance. There are many aspects of dynamical networks that require statistical considerations. In this paper we focus on determining network structure. Estimating dynamic networks is a difficult task since the number of components involved in the system is very large.
As a result, the number of parameters to be estimated is bigger than the number of observations. However, a characteristic of many networks is that they are sparse. For example, the molecular structure of genes make interactions with other components a highly-structured and therefore sparse process.

Penalized Gaussian graphical models have been used to estimate sparse networks. However, the literature has focussed on static networks, which lack specific temporal constraints. We propose a structured Gaussian dynamical graphical model, where structures can consist of specific time dynamics, known presence or absence of links and block equality constraints on the parameters. Thus, the number of parameters to be estimated is reduced and accuracy of the estimates, including the identification of the network, can be tuned up. Here, we show that the constrained optimization problem can be solved by taking advantage of an efficient solver, logdetPPA, developed in convex optimization. Moreover, model selection methods for checking the sensitivity of the inferred networks are described. Finally, synthetic and real data illustrate the proposed methodologies.

\end{abstract}

\section{Introduction}

Graphical models are powerful tools for analyzing relationships between a large number of random variables. Their fundamental importance and universal applicability are due to two main factors. Firstly, graphs can be used to represent complex dependence relationships among random variables. Secondly, their structure is modular which means that complex graphs can be built from many simpler graphs. 	

A graph consists of a set of nodes and a set of links between these nodes. In a graphical model the nodes are associated with random variables and links represent conditional dependence between the nodes. In a Gaussian graphical model (GGM) it is assumed that these random variables follow a multivariate normal distribution. An important property of GGMs is that the concentration matrix, i.e. the inverse of the covariance matrix, fully determines represents the conditional independence structure, i.e. the graph.

The maximum likelihood estimator (MLE) for the concentration matrix is the inverse of the sample covariance matrix, if it exists, and it exits when the number of observations is greater than the number of random variables. However, in many modern applications experiments consists of observing many features and many fewer observations. In genetics, variables are typically thousands of genes or gene products, whereas there are at best only a few hundred of observations. Moreover, genetic networks are sparse, which means many of the genes behave conditionally independent from the others. In terms of GGMs, it means that most of the elements in the precision matrix are equal to zero \citep{jeong2001lethality}. As \cite{dempster1972covariance} pointed out, parameter reduction involves a trade-off between costs and benefits, i.e. the amount of noise in a fitted model due to errors of estimation is reduced but errors of miss-specification are introduced because the null values are incorrect. \cite{dempster1972covariance} introduced sparse covariance matrix estimation. The subset of zeros in the precision matrix can be selected in various ways analogous to the various forward and backward procedures used for selecting predictors variables in multiple regression analysis. However, the forward and backward procedures are not suitable for high-dimensional data problems.

$\ell_1$-penalized likelihood inference \citep{tibshirani1996regression}, which has been extensively applied in regression models, can be adapted to graphical models in order to estimate sparse graphs. The idea is to constrain the $\ell_1$-norm penalty, i.e. the sum of the absolute values of the inverse of the covariance matrix, to be less or equal to a tuning parameter. The smaller the tuning parameter, the more zeroes will be estimated in the precision matrix. \cite{meinshausen2006high} proposed to select edges for each node in the graph by regressing the variable
on all other variables using $\ell_1$ penalized regression. This method reduces to solving $p$ separate regression problems, and does not provide an estimate of the matrix itself. Penalized maximum likelihood approaches using the $\ell_1$  penalty have been considered by \cite{yuan2007model, banerjee2008model, friedman2008sparse, rothman2008sparse}, who have all proposed different algorithms for computing the estimator of the precision matrix.  \cite{fan2001variable} introduced clipped absolute deviation penalty (SCAD). On the other hand, \cite{lam2009sparsistency} extended this penalized maximum likelihood approach to general non-convex penalties. Alternative penalized estimators based on the pseudo-likelihood instead of the likelihood were recently proposed by \cite{peng2009partial}. Theoretical properties of the $\ell_1$ penalized maximum likelihood estimator in the large $p$ scenario were derived by \cite{meinshausen2006high} as well as \cite{rothman2008sparse}. \cite{lam2009sparsistency} established a so called \lq\lq{}sparsistency\rq\rq{} property of the penalized likelihood estimator, which means that all parameters that are zeroes are estimated as zero with probability tending to one when the sample size increases.

The complexity of the GGMs can be reduced if one imposes some symmetry
constraints on the precision matrix, i.e. the number of parameters to be estimated is reduced by introducing equality constraints. Recently, \cite{højsgaard2008graphical} proposed Gaussian graphical models with symmetry constraints. An important motivation for considering structured GGMs is that conditions for maximum likelihood estimates to exist are less restrictive than for standard GGMs. As \cite{højsgaard2008graphical}  pointed out symmetry models are particularly useful when parsimony is needed, for example when estimating precision matrices of large dimension with relatively few observations.

The idea of this paper is to combine symmetry models and graphical lasso for modelling dynamic graphical models. In particular, we propose two models that impose symmetries constraints on the precision matrix $\Theta$ and on the conditional correlation matrix $\Omega$, which is the negative scaled concentration matrix. The latter model is useful when the random variables are not measured on the same scale. We structure graphical models in a way that comes naturally to time-course data, by using the idea of coloured graphs to define factorially-structured precision matrices.

In Section \ref{sec:motivating-examples}, we introduce two motivating examples: a time-course
microarray experiment involving T-cells and an educational experiment involving the Edu dataset.
In Section \ref{sec:fgl}, we introduce the idea of coloured graphs for time-course datasets
and Gaussian graphical models. The factorial graphical lasso for both precision, FGL$_\Theta$, and conditional correlation matrices, FGL$_\Omega$, are discussed in detail and the convex optimization problem with linear constraints is explained in section \ref{sec:pmle}. Section \ref{sec:modelselection} addresses model
selection and parameter smoothing selection which are important issues in graphical lasso. In particular, classical approaches such as AIC, BIC are derived to do model selection, and stability selection \citep{meinshausen2010stability} has been adapted to factorial graphical lasso models. Section \ref{sec:simulationstudy} provides encouraging numerical results in a simulation study. In Section \ref{sec:dataanalysis}, T-cell and Edu datasets are studied in full detail.

\section{Motivating examples}
\label{sec:motivating-examples}
An important issue in system biology is to understand the system of interactions among several biological components such as protein-protein interaction and gene regulatory networks. Hence, several techniques have been developed to collect data from different organisms. For instance, microarrays measure gene expression levels, i.e. the concentration of messenger RNA produced when the gene is transcribed. 	

A single microarrays is a snapshot in time of the expression of genes of an organism. Gene expression, therefore, is a temporal process, which evolves dynamically in response to internal, genomic, and external, environmental, cues.  Even under stable conditions, mRNA is transcribed continuously and new proteins are generated. This process is highly regulated. In many cases, the expression programme starts by activating a few transcription factors, which in turn activate many other genes that act in response to the new state. Transcription factors are proteins that bind to specific DNA sequences, thereby controlling the flow, i.e. transcription, of genetic information from DNA to mRNA. Taking a snapshot of the expression profile following a new condition can reveal some of the genes that are specifically expressed in the new state. But rather than determining the set of differentially expressed genes, such as in the early days of microarray experiments, biologists are more interested in determining the transcriptional programme, i.e. determining the functional pathways. In order to infer the temporal interaction between the genes, it is necessary to perform time-course expression experiments. The time-course genetic T-cell dataset \citep{rangel2004modeling} is described in subsection \ref{microarraydata}.

Many educational studies aim to establish the effect of one or more variables on one or more classroom outcomes. Often such observational studies collect a large number of variables on a large number of disparate scales. Often various proxy measures are available for these variables and therefore it is not uncommon to expect a complex set of relationships between the variables. Our second motivating example is a longitudinal educational study described in \cite{opdenakkerchanges2012, opdenakker2011teacher}. We describe this example in subsection \ref{edudata}. The aim is to find a network that describes the relations among the items measured in the experiment. This second example clarifies the difference between the use of $FGL_\Theta$ and $FGL_\Omega$.

\subsection{Example 1: Human T-cell microarray data}
\label{microarraydata}

Two cDNA microarray experiments were performed to collect gene expression levels for T-cell activation analysis. Activation of T-cell was produced by stimulating the cells with two treatments: the calcium ionosphere and the PKC activator phorbol ester PMA. The human T-cells coming from the cellular line Jakart were cultured in a laboratory. When the culture reached a consistency of $10^6$ cells/ml, the cells were treated with the two treatments PMA and PKC. Gene expression levels for 88 genes were collected for the following times after the treatments: 0, 2, 4, 6, 8, 18, 24, 32, 48, 72 hours.
In the first experiment the microarray was dived such that 34 subarray were obtained. Each of these 34 subarray contained the strands of the 88 genes under investigation. Strands are the complementary base for the rRNA which is the transcribed copy of a single strand of DNA after the process of transcription. In the second microarray experiment the microarray was dived into 10 sub-arrays.  Each of these 10 sub-arrays contained the strands of the 88 genes under investigation. Each microarray was composed by ten different slides which were used for the two experiments to collect temporal measurements.  For time 0 a set of cells were hybridized to the first slide after the cell cultured reached the right density and before the treatments were applied. For the second time point (time 2), another set of cells were hybridized in a second slide. The experiment was conducted by \citep{rangel2004modeling}.

At this point we assume that the technical replicates are independent samples, and that the temporal replicates are dependent replicates from the same samples. These two assumptions result in a dataset with 44 independent replicates. These are strong assumptions, but can be justified from the underlying sampling scheme. Nevertheless, this means that the conclusions from the analysis on T-cell, shown in section \ref{sec:dataanalysis}, should be critically considered. Two further steps were conducted by \cite{rangel2004modeling} to obtain a set of genes that were highly expressed and normalized across the two microarrays. Firstly, thirty genes with high variability between the two microarrays and within the same time point were removed.  Secondly, normalization methods were applied to remove systematic variation due to experimental artifacts. The normalization method used by \cite{rangel2004modeling} is described in the paper written by \cite{bolstad2003comparison}. This results in a 44 i.i.d. $10\times 58$ dimensional observations. 

In this subsection, we describe a possible partition for the maximum likelihood estimator of the precision matrix that is the inverse of the sample variance-covariance matrix. For illustrative purposes, we select a subset of 4 genes $\Gamma = \{ZFN, CGN, SIV, SCY \}$ across 2 time points $T = \{1, 2 \}$. Then,  we compute and show the empirical concentration matrix $\mathbf{S}^{-1}$ for the selected subset based on 44 observations in Table \ref{concentrationo}. In particular, $\mathbf{S}^{-1}$ can be partitioned as follows,
\begin{eqnarray}
\label{eq1}
\mathbf{S}^{-1} &=& \mathbf{I}_{\Gamma T}\circ \mathbf{S}^{-1} +
\left( \begin{array}{rl} \mathbf{0} & \mathbf{I}_{\Gamma \Gamma}\\ \mathbf{I}_{\Gamma \Gamma} & \mathbf{0} \\ \end{array} \right) \circ \mathbf{S}^{-1}  \nonumber \\
&& \left( \begin{array}{rl} \mathbf{D}_{\Gamma \Gamma} &\mathbf{0} \\ \mathbf{0} &\mathbf{D}_{\Gamma \Gamma} \\ \end{array} \right) \circ \mathbf{S}^{-1} +
\left( \begin{array}{rl} \mathbf{0} & \mathbf{D}_{\Gamma \Gamma} \\ \mathbf{D}_{\Gamma \Gamma} & \mathbf{0} \\ \end{array} \right)\circ \mathbf{S}^{-1},
\end{eqnarray}
where $\circ$ is the elementwise Hadamard matrix product, $\mathbf{I}_{\Gamma T}$,  $\mathbf{I}_{\Gamma \Gamma}$ are identity matrices, and $\mathbf{D}_{\Gamma \Gamma}$ is a square matrix with ones off the diagonal.

At this point we notice that four terms are showed in the summation and each of these terms can be interpreted. The first term indicates how well the variance of $(gene_{ij})_{i\in \Gamma,j \in T}$ is predicted given the rest. For example it is easier to predict the behaviour of gene SIV than the behaviour f gene SCY. The second shows self-self conditional independence at temporal lag 1. The self-self conditional independence represent relationships between the same genes measured into different time points. For example, gene ZNF is conditional dependent with itself given the rest of the genes and the dependence seems to be very strong $-0.93$. So, if gene ZNF is upregulated at time 0 then it will be upregulated at time 1. The third term indicates the conditional dependencies at time 0 and 1 between genes at temporal lag zero. For the t-cell experiment this means that one is considering the relations between genes before the effect of the treatment. For example, gene CCN and SIV can be considered conditionally independent ($0.02$ and $0.06$) since for both the temporal lag the element of the precision matrix is small. The last term indicates the conditional dependencies between time 0 and 1 at lag zero. In Section \ref{sec:fgl} we will refer to these partitions of the concentration matrix as natural partitions. Moreover, we will use the idea of coloured graphs to make a different partition for the different temporal lags.

\begin{table}
\caption{\label{concentrationo}Maximum likelihood estimator of the precision matrix or empirical concentration matrix $\mathbf{S}^{-1}$ based on 44 replicates for 4 genes measured across 2 time points. The number of genes measured in the T-cell experiment was 58 across 10 time points but we randomly selected four genes to give an illustration of the estimated precision matrix. The lower part of the matrix can be read from the upper part since the matrix is symmetric.}
\centering
\fbox{%
\begin{tabular}{rrrrrr|rrrr}
 Time &\multicolumn{5}{c}{ \em 1} &\multicolumn{4}{c}{ \em 	2}  \\
 &  Gene & \em ZNF &\em CCN &\em SIV &\em SCY & \em ZNF &\em CCN &\em SIV &\em SCY \\
  \hline
\multirow{4}{*}{\em 1} & \em ZNF & 1.38 & -0.05 & -0.50 & 0.25 & -0.20 & -0.12 & -0.01 & -0.11 \\
&\em CCN & - & 1.58 & 0.02 & -0.39 & -0.12 & -0.93 & -0.02 & 0.07 \\
&\em SIV & - & - & 1.61 & -0.13 & -0.17 & 0.12 &-0.78 & -0.02 \\
&\em SCY & - & - & - & 1.29 & 0.05 & 0.25 & 0.49 & -0.09 \\
\hline
\multirow{4}{*}{\em 2} & \em ZNF  & - & - & - & - & 1.10 & -0.18 & 0.04 & 0.08 \\
&\em CCN  & - & - & - & - & - & 1.78 & 0.06 & 0.42 \\
&\em SIV & - & -& - & - & - & -& 1.58 & 0.09 \\
&\em SCY & - & - & - & - & - & - &- & 1.16 \\
\end{tabular}}
\end{table}

\subsection{Example 2: Educational study.}
\label{edudata}

The education dataset is a longitudinal dataset in which a set of scores regarding teacher and student behaviour were collected across 5 different time points, to wit $0, 1, 4, 7, 10$ months. The study was conducted in the Netherlands and the students were followed-up during their first year of the secondary school \citep{opdenakkerchanges2012, opdenakker2011teacher}. In particular, 20 classes of students and 24 educational scores were considered.

For illustrative purposes, we consider a subset of four scores by naming the variables: Controlled (Contr), Autonomies (Auton), Influence (Infl) and Proximity (Prox). The first two scores concern student behaviour, taking values between -3 and 3. The last two scores measure teacher behaviour and take values between 0 and 5. The scale for these variables is not compatible and conditional correlations are therefore more meaningful than the concentrations. In section \ref{sec:fgl}, we show that conditional correlations are invariant under changes of scale for individual variables. The empirical conditional correlations based on 20 classes of students for 4 scores measured across 2 time points are shown in the upper triangular block (italic numbers) of the matrix in Table \ref{correlationo} while the conditional covariances are shown in the lower triangular block.

\begin{table}
\caption{\label{correlationo}Empirical conditional correlations (upper triangular and diagonal) and conditional covariance (lower triangular) based on 20 replicates for 4 score measures measured across 2 time points. The number of educational measures in the experiment was 24 across 4 time points but we randomly selected four measures to give an illustration of the estimated precision matrix and the estimated scaled precision matrix.}
\centering
\fbox{%
\begin{tabular}{rrrrrr|rrrr}
 Time &\multicolumn{5}{c}{ \em 1} &\multicolumn{4}{c}{ \em 	2}  \\
& Subject& \em Contr & \em Auton & \em Infl & \em Prox & \em Contr & \em Auton & \em Infl &\em Prox \\
  \hline
\multirow{4}{*}{\em 1} & \em Contr& 1.79  & \em 0.31 & \em -0.07 &\em 0 &\em 0.61 &\em -0.23 &\em 0.18 &\em -0.06 \\
& \em  Auton &-0.57 & 1.86  & \em 0.03 &\em 0.22 &\em -0.19 &\em 0.57 &\em -0.07 &\em -0.03 \\
& \em Infl &  0.11 & -0.05 & 1.25  &\em -0.04 &\em 0 &\em -0.01 &\em 0.41 &\em 0.09 \\
& \em  Prox & 0.01 & -0.37 & 0.05 & 1.55 &\em 0.04 &\em -0.03 &\em -0.12 &\em 0.51 \\
  \hline
\multirow{4}{*}{\em 2} & \em   Contr2 &  -1.09 & 0.35 & 0 & -0.07 & 1.76  &\em 0.31 &\em -0.08 &\em -0.01 \\
& \em   Auton & 0.43 & -1.09 & 0.02 & 0.05 & -0.58 & 1.96  &\em 0.13 &\em 0.21 \\
& \em Infl & -0.28 & 0.11 & -0.53 & 0.18 & 0.12 & -0.21 & 1.37  &\em 0.21 \\
& \em Prox & 0.10 & 0.06 & -0.13 & -0.82 & 0.02 & -0.39 & -0.32 & 1.68  \\
\end{tabular}}
\end{table}

Similarly the example in the preview Subsection \ref{microarraydata} we can partition the empirical conditional correlation matrix in 4 parts as follows:
\[
\mathbf{R}^{-1} = -\mathbf{I}_{\Gamma T} +
\left( \begin{array}{rl} \mathbf{0} & \mathbf{I}_{\Gamma \Gamma}\\ \mathbf{I}_{\Gamma \Gamma} & \mathbf{0} \\ \end{array} \right) \circ \mathbf{R}^{-1}
\left( \begin{array}{rl} \mathbf{D}_{\Gamma \Gamma} &\mathbf{0} \\ \mathbf{0} &\mathbf{D}_{\Gamma \Gamma} \\ \end{array} \right) \circ \mathbf{R}^{-1} +
\left( \begin{array}{rl} \mathbf{0} & \mathbf{D}_{\Gamma \Gamma} \\ \mathbf{D}_{\Gamma \Gamma} & \mathbf{0} \\ \end{array} \right)\circ \mathbf{R}^{-1},
\]	
were $R^{-1}$ is equal to the negative of the scaled precision matrix. This means that the conditional elements on the lower part have been divided by the square root of the two corresponding elements on the diagonal $r_{ij} = -s_{ij}/\sqrt{s_{ii} s_{jj}}$. At this point each of diagonal elements is in the range $-1$ and $1$ and it corresponds to a conditional correlation measure. The interpretation of these measures is equivalent to the interpretation of the elements of the precision matrix. However, conditional correlations are pure numbers, i.e. the scale of the variables do not have any influence, and they are in a range between $-1$ and $1$ so that numbers close to zero indicates conditional independence, as before, and numbers close to one or -1 indicate strong positive and negative conditional dependencies, respectively. Whereas of the conditional covariance elements a positive conditional correlation indicate a positive relations between the corresponding genes.

\section{Factorial graphical model for dynamic networks}
\label{sec:fgl}

An undirected graph $G = (V, E)$, where $V$ is a set of vertices corresponds to a set of random variables, and $E$ the set of links corresponds to a set of conditional independencies, is defined a Gaussian graphical model if the set of random variables follow a Gaussian distribution. Vertices of the graph $G$ in static dataset have no natural order. Whereas for longitudinal dataset or time-course dataset a natural order can be established. For example, gene ZNF in the T-cell dataset is measured at time 1 (which is gene expression are collected after 0 hour of treatment) and at time 2 (which is after two hours of treatment). Vertexes ZNF, CCN, SIV and SCY are called natural vertexes. Generally speaking, let $G = (V, E)$ be a graph where $V = (v_{jt})_{j \in \Gamma, t \in T}$ is a finite set of vertices and $E \subseteq V \times V$ is a subset of ordered pairs of distinct vertices. The set of natural vertices is $\Gamma = \{1,\ldots, n_{\Gamma} \}$, where $T = \{1,\ldots, n_T\}$ is an ordered set, typically describing time points. Note that the set $\Gamma$ is unlabelled, but here described as an ordered set for convenience of notation. Moreover, we have seen in example 1 and 2 that the empirical correlation matrix or the empirical concentration matrix can be written as a sum of four terms. These four terms are called natural partitions, and they can be represented in a dynamic graph as shown in Figure \ref{figgrahex}. Here, we assume that the last term of formula \ref{eq1} is zero so that genes at time 1 and at time 2 are assumed conditionally independent. The first term of formula \ref{eq1} is not represented in Figure \ref{eq1}. The second term is represented by the dot lines and it represents the self-self interactions. The third term is represented by the dashed and continuous lines.

\begin{figure}
\centering
 \makebox{\includegraphics[scale = 1]{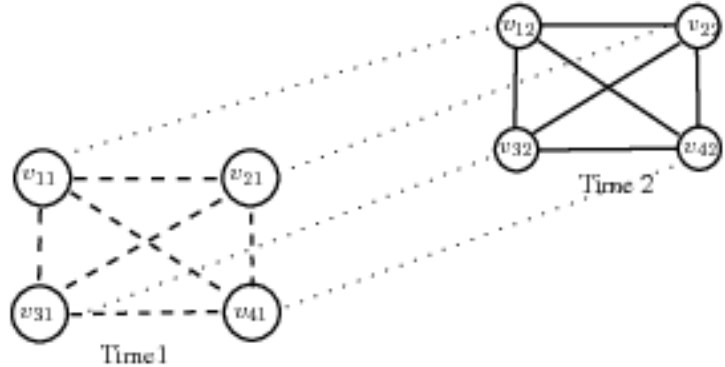}}
\caption{\label{figgrahex}Example of a dynamic coloured graph with four natural vertices measured across two time points. The natural partitions of equation \ref{eq1} have been represented with different type of the lines. The first term is not represented, The second term is represented by dot lines. The third term is represented by dashed and continuous lines.}
\end{figure}

We will refer to the graph in Figure  \ref{eq1} as dynamic graphs, which can be formally defined as follows:
\begin{definition}[Dynamic graphs]
A  dynamic graph is a pair $G = (V, E)$, where $V = \left\{v_{ij}\right\}_{i\in \Gamma, j\in T}$ is a finite set of vertices and $E \subseteq V \times V$ is a set of ordered couples of elements.  $\Gamma$ and $T$ are finite sets.
\end{definition}
Hence, the main characteristic of dynamic graphs is that the same vertices are measured across different time points. In the T-cell example the same 56 genes were measured across 10 time points and between time point $t$ and $t+1$ the experimental conditions were changed. 
We notice that conditional correlations between the natural vertices at time point one and time point 2 described in Table \ref{correlationo} are similar. On the end the diagonal elements in the first block of Table \ref{concentrationo} and in the second block are similar. Moreover, it is intuitive to think that the conditional correlation between gene 1 at time 1 and gene 1 at time 2 can be considered equivalent or that self-self correlations between genes across two time points can be consider to be equal. In order to formalize this idea we use of coloured graphs.

\begin{definition}[Coloured Graph]
A  coloured graph $\tilde{G} = (V,  E, F)$ is a triplet, where
$G = (V,E)$ is a graph and $F$ is a mapping on the links, i.e.
\[F:E\longrightarrow C,\]
where $C$ is a finite set of colours.
\end{definition}
Thus, a coloured graph induces partitions on the graph by $F$, that is different subsets of $E$ are visualized with different colours. For example, three different colours (i.e, three different type of lines) are represented in Figure \ref{figgrahex}. The dash line represents one colour (or sub-partition), the continuous line represents the second colour and the dot line represents the third colour.
Let us denote the partitions induced by the coloured graphs $F$ by $E \prec F$, i.e.:
$$E \prec F = \cup_{c \in C} F^{-1} \cap E,$$
where $E \prec F $ indicates that the partition is induced on $E$, and $E$ stands for the complete set of links. The mapping can be applied to the subsets of $E$ which are $S_i$ and $N_i$. Each partition represents relationships between natural vertices at some time point $t \in T$.

Let $\{S_i\}_{i = 0}^{n_T -1}$, and $\{N_i\}_{i = 0}^{n_T -1}$ be subsets of vertices and links where $S_i, N_i$ are natural partitions such that:
\[S_i = \{ \{(v_{jt}, v_{j,t+i}), (v_{j,t+i}, v_{jt})\} | j \in \Gamma, t = 1, \ldots, n_T-i\},\]
and
\[N_i = \{ \{(v_{jt}, v_{k,t+i}),(v_{k,t+i}, v_{jt})\} | \forall j \neq k \in \Gamma, t = 1, \ldots, n_T-i\}.\]
Each of these partitions is interpreted as follows: $S_i$ considers the lag $i$ interactions between the same natural vertices, and $N_i$ is a graph at time lag $i$. We induce further partitions on $S_i$ and $N_i$ by using the idea of coloured graphs in order to give more consistent interpretations of dynamic graphs. In particular, we consider four mappings. Firstly, $E \prec F_1$ indicates that all edges in the partition are coloured with the same colour. Secondly, $E \prec F_T$ indicates that all edges in the partition are coloured with colours which are the same within natural vertices. Thirdly, $E \prec F_{\Gamma}$ indicates that all edges in the partition are coloured with the same colour within time points and different colours across natural vertices. Finally, $E \prec F_{\Gamma T}$ indicates that all edges in the partition are coloured differently across time points and natural vertices. This can be summarize with the following functions:
\begin{eqnarray*}
F_1:&& E \rightarrow C\\ 
&&\forall v_i,v_j \in E: ~~F_1(v_i) = F_1(v_j)\\
F_T:&& E  \rightarrow C\\
 &&\forall v_{it,js}, v_{ku, lv} \in E, \mbox{if } t=u  \mbox{ and }  s=v:\\
 &&F_T(v_{it,js}) = F_T(v_{ku, lv})\\
F_{\Gamma}:&& E  \rightarrow C\\ 
&&\forall v_{it,js}, v_{ku, lv} \in E, \mbox{if }  i = k \mbox{ and }  j = l:\\
&&F_{\Gamma}(v_{it,js}) = F_{\Gamma}(v_{ku, lv})\\
F_{\Gamma T}:&& E  \rightarrow C\\
 &&\mbox{no restriction}
 \end{eqnarray*}
where we can substitute at $E$ a specific natural partition $S_i$ or $N_i$ for $i = 1, \ldots, T-1$. Thus, $S_i \prec F_i$ and $N_i \prec F_j$, where $i,j = F_1, F_T, F_{\Gamma}, F_{\Gamma T}$ induce partitions as follows:
$$S_i = \{S_i^m\}_{m= 1}^{s_i},$$
and
$$N_i =\{N_i^m\}_{m=1}^{n_i}.$$
For example, $S_i \prec F_T$ induces the following partition of $S_i$:
\[S_i=\{S^1_i \cup \ldots \cup S^{n_T-i}_i \},\]
where $S^t_i = \{\{(v_{jt}, v_{j,t+i}), (v_{j,t+i}, v_{jt})\} | j \in \Gamma, t = 1, \ldots, n_T-i\}$. Edges belonging to $S^t_i$ are
coloured with the $n_{T}-i$ colours $C = \{C_1, \ldots, C_{n_T-i} \}$, respectively. We abuse notation and let $N_i \prec 0$ and $S_i \prec 0$ in order to express that $N_i = \{\emptyset\}$ and $S_i = \{\emptyset$\}, i.e. the graph $G$ does not contain such edges.

In Table \ref{degreeoffreedom} we show the number of colours for any combination of $S_i, N_i$ and graph colouring.
\begin{table}
\caption{\label{degreeoffreedom}Number of colours for any combination of $S_i, N_i$ and graph colouring. $S_i, N_i$ represent natural partitions of $E$, where $E$ is a set of links. The natural partitions are sub partitions of the set $E$.}
\centering
\fbox{%
\begin{tabular}{ccccccc}
Factor&$S_0$&$N_0$&$S_1$&$N_1$&$S_2$&$N_2$\\
\hline
$F_1$& 1&1  &1  &1  &1   &1\\
$F_T$ &$n_{T} $&$n_{T}$&$n_{T}-1$&$n_{T}-1$&$n_{T}-2$&$n_{T}-2$\\
$F_{\Gamma}$ & $n_{\Gamma} $& $\frac{1}{2}n_{\Gamma}(n_{\Gamma}-1)$ & $n_{\Gamma} $& $n_{\Gamma}(n_{\Gamma}-1)$& $n_{\Gamma} $& $n_{\Gamma}(n_{\Gamma}-1)$\\
$F_{\Gamma T}$ & $n_{\Gamma}n_{T} $& $\frac{1}{2}n_{\Gamma}(n_{\Gamma}-1)n_T$ &$ n_{\Gamma}(n_{T}-1)$& $n_{\Gamma}(n_{\Gamma}-1)$ & $n_{\Gamma}(n_{T}-2)$ & $n_{\Gamma}(n_{\Gamma}-1)$\\
 &  &    &  &$\times (n_{T}-1)$  & & $\times (n_{T}-2)$\\
\end{tabular}}
\end{table}
The total number of colours $n_C$ can be calculate from Table \ref{degreeoffreedom}.
Figure \ref{figgrahex} shows an example of  ``coloured graph'' where vertices $(v_{ij})_{i \in \Gamma, j \in T}$ are all of the same colour and edges with the same line styles are of the same colours. The graph resembles the following model:
\begin{equation}
[S_0 \prec F_1, N_0 \prec F_T, S_1 \prec F_1, N_1 \prec 0],
\label{eq:modelprec}
\end{equation}
where, firstly the following natural partitions (sub-partitions) are created:
\begin{eqnarray}
\nonumber
S_0 &=&\{v_{11}, v_{12}, \ldots, v_{pT}\},\\
\nonumber
S_1&=& \{(v_{11}, v_{12}), (v_{21}, v_{22}), (v_{31}, v_{32}), (v_{41}, v_{42})\}, \\
\nonumber
N_0 &=& \{(v_{11}, v_{21}),  (v_{21}, v_{41}), (v_{41}, v_{31}), (v_{31}, v_{11}),  (v_{11}, v_{41}), (v_{21}, v_{31}),\\
\nonumber
&=& (v_{12}, v_{22}),  (v_{22}, v_{42}), (v_{42}, v_{32}), (v_{32}, v_{12}),  (v_{12}, v_{42}), (v_{22}, v_{32})\}.
\end{eqnarray}
Note that if a couple $(v_{ij}, v_{kl})$ is present then the couple $(v_{ij}, v_{kl})$ is present too. We have omitted the symmetric couples from the sets to simplify the notation. Secondly, $S_0 \prec F_1$ induces the following subset:
$$S^1_0 = \{v_{11}, v_{12}, \ldots, v_{pT}\},$$
so that a colour is created. Then, $N_0 \prec F_T$ brings the following two sub-partitions:
\begin{eqnarray}
\nonumber
N^1_0 &=&  \{(v_{11}, v_{21}),  (v_{21}, v_{41}), (v_{41}, v_{31}), (v_{31}, v_{11}),  (v_{11}, v_{41}), (v_{21}, v_{31})\},\\
\nonumber
N^2_0 &=&  \{(v_{12}, v_{22}),  (v_{22}, v_{42}), (v_{42}, v_{32}), (v_{32}, v_{12}),  (v_{12}, v_{42}), (v_{22}, v_{32})\},
\end{eqnarray}
so that two colours have been created.

In order to consider a graphical model nodes of the dynamic graph need to be associated with the random variables. Then, we can assume that nodes follow a multivariate normal distribution. In what follows, we consider the set of vertices and nodes of the dynamic coloured graph. Moreover, we denote with $\mathbf{Y} = (Y_{v_{ij}})_{v_{ij} \in V} \in \mathbb{R}^{\Gamma T}$ the set of random variables. Each vertices $V = (v_{ij})_{i \in \Gamma,j \in T}$ is related to a random variable $Y_{v_{ij}}$.
\begin{definition}[Dynamic graphical models]
A graphical model $M = (G, \mathbb{P})$ is a couple $G$ and $\mathbb{P}$, where $G$ is a dynamic graph and $\mathbb{P}$ is a probability distribution on $\mathbf{Y}$ satisfying some Markovian properties, i.e. set of conditional independence relations encoded by the undirected edges $E$.
\end{definition}

In order to connect the idea of coloured graphs and graphical models we give the following definition:
\begin{definition}[Factorial dynamic graphical models]
A factorial graphical model $M = (G, \mathbb{P}, F)$ is a triplet $(G, \mathbb{P}, F)$, where $G$ is a dynamic graph, $\mathbb{P}$ is a probability distribution and $F$ is a mapping on the natural partitions $S_i$ and $N_i$, for $i= 1, \ldots, T-1$ and $T$ is the total number of time points.
\end{definition}
Assume that $\mathbf{Y} \sim N(\boldsymbol \mu, \boldsymbol \Sigma)$ follows a multivariate normal distribution then we have a factorial Gaussian graphical models where pairwise conditional independencies are equivalent to zeros in the concentration matrix $\boldsymbol \Theta = \boldsymbol \Sigma^{-1}$, i.e:
\[Y_{v_{ij}} \perp Y_{v_{kl}}|\mathbf{Y}_{V\setminus\{v_{ij},v_{kl}\}} \Leftrightarrow \theta_{\{ij,kl\}}=0.\]
The scaled off diagonal elements $\omega_{ij,kl|V \setminus\{ij, kl\}}  = -\frac{\theta_{ij,kl}}{\sqrt{\theta_{ij,ij}\theta_{kl,kl}}}$ are the negative of conditional correlation coefficients, where $i, k \in \Gamma$ and $j, l \in T$ for $i \neq j$ and $k \neq l$. For future use, we denote $\boldsymbol \Omega = (\omega_{ij,kl|V\setminus \{ij,kl\}})$ be a matrix of scaled elements of $\boldsymbol \Theta$.

For example, the factorial Graphical model with graph represented in Figure \ref{figgrahex}, $\mathbf{Y} = (Y_{11}, Y_{12}, \ldots, Y_{42})$ vector of random variables which correspond to vertices $\{v_{11}, v_{12}, \ldots, v_{42}\}$ and relations $[S_0 \prec F_1, N_0 \prec F_T, S_1 \prec F_1, N_1 \prec 0]$ given in (\ref{eq:modelprec}) imply the following precision matrix:
\[ \boldsymbol \Theta = \left[ \begin{array}{ccc|ccc}
 \theta_1 & \theta_2   &   \theta_2      &    \theta_3    &  0 & 0        \\
 \theta_2      & \theta_1 & \theta_2  &   0      &  \theta_3 & 0   \\
 \theta_2  &   \theta_2    & \theta_1 &  0       &     0    & \theta_3        \\
 \hline
 \theta_3  &0   &  0 & \theta_1 &   \theta_4     &    \theta_4      \\
   0      & \theta_3       &  0 & \theta_4   & \theta_1 & \theta_4   \\
   0     &    0    & \theta_3 &  \theta_4   &    \theta_4  &   \theta_1
  \end{array} \right]. \]

We have described the idea of coloured graphs which allows us to create sub-partitions of natural partitions. Moreover, equality constraints on conditional correlations can be imposed such that edges with the same colours imply these equality restrictions. Here, we define a set of design matrices $\mathbf{X}$ which is useful to directly connect elements of a coloured graph with concentration or conditional correlation matrix parameters.

A coloured graph that defines partitions on $E$, i.e. $\{S^m_i\}$ and $\{N^m_i\}$, can be associated with two sets of design matrices $\mathbf{X}^S = \{\mathbf{X}^{S^m_i}\}_{i = 0, m= 1}^{n_{\Gamma-1}, s_i}$ and $\mathbf{X}^N = \{\mathbf{X}^{N^m_i}\}_{i = 0, m = 1}^{n_{\Gamma-1},n_i}$ where $\mathbf{X}^{S^m_i}, \mathbf{X}^{N^m_i} \in \mathbb{R}^{\Gamma T}$ and $\mathbf{X}^{S^m_i}$ can be uniquely identified such that
$$x^{S^m_i}_{jt,gs} =  \left\{ \begin{array}{rl}
 1 &\mbox{ if }  (v_{jt}, v_{gs}) \in S^m_i\\
 0 &\mbox{ otherwise}
       \end{array} \right.
$$
and
$$x^{N^m_i}_{jt,gs} =  \left\{ \begin{array}{rl}
 1 &\mbox{ if } (v_{jt}, v_{gs}) \in N_i^m \\
 0 &\mbox{ otherwise}
\end{array} \right.
$$
We re-define the set of design matrices as:
\[
\mathbf{X} = \{\mathbf{X}^S, \mathbf{X}^N\} = \{\mathbf{X}_1, \ldots, \mathbf{X}_{n_c}\},
\]
where $\mathbf{X}^S = \sum_{i=1}^{n_S} \mathbf{X}^{S_i^m}$, $\mathbf{X}^N_i = \sum_{i=1}^{n_N} \mathbf{X}^{N_i^m}$, $n_c = n_S + n_N$ is the total number of colours, $n_S$ is the total number of colours for the self-self partitions and $n_N$ is the total number of colours for the network partitions.

\subsection{Structured Gaussian graphical models for conditional concentration matrix}

The design matrix $\mathbf{X}$ can be used to induce the following parametrization on $\boldsymbol \Theta$, i.e.
\[\boldsymbol \Theta = \sum^{n_C}_{m = 1} \mathbf{X}^{(m)} \theta_m.\]
where $\boldsymbol \theta = (\theta_m)_{m = 1}^{n_C}$ is a vector of unknown parameters. Because $\mathbf{X}_m$ are induced by the graph the specific elements of the concentration matrix which correspond to edges with the same colours are constraint to be equal.
Note that the restrictions so defined are linear in the concentration matrix.

\paragraph{Example 1: Human T-cell.}Consider the following factorial graphical model for human T-cell microarray data (see Section \ref{sec:motivating-examples} for a description):
$$[S_0 \prec F_1, N_0 \prec F_1, S_1 \prec F_1, N_1 \prec 0],$$
then the design matrices $\mathbf{X}= \{\mathbf{X}^S,\mathbf{X}^N \}= \{\mathbf{X}_1,\mathbf{X}_2,\mathbf{X}_3\}$ are:
\[
\mathbf{X}^{S^1_0} = \mathbf{I},  \mathbf{X}^{S^1_1} =  \left( \begin{array}{rl} \mathbf{0} &\mathbf{I}\\ \mathbf{I} &\mathbf{0} \\ \end{array} \right), \mathbf{X}^{N^1_0} = \left( \begin{array}{rl} \mathbf{D} &\mathbf{0} \\ \mathbf{0} &\mathbf{D} \\ \end{array} \right), \mathbf{X}^{N^1_1} = \left( \begin{array}{rl} \mathbf{0} &\mathbf{0} \\\mathbf{0} &\mathbf{0} \\ \end{array} \right)
\]
where $\mathbf{D}$ is a square matrix with 1 off-diagonal and 0 on the diagonal. Note that $\mathbf{X}^{N^1_1}$ is an empty matrix since $S^1_1$ is an empty set which implies 0 colours. Table \ref{estconcentration} shows the estimated concentration matrix based on 44 replicates for 4 genes measure across 2 time points. The model induces $n_C = 3$ colours. For the estimation procedure see \cite{højsgaard2008graphical}. For a faster and more general estimation procedure see the algorithm in subsection \ref{sec:Penalized likelihood for coloured graphs}.

\begin{table}
\caption{\label{estconcentration}Estimated conditional covariance based on 44 replicates for 4 genes measured across 2 time points.  The number of parameter estimated is 3, that is the number of colour in the graph is 3. Note that the following model $S_0 \prec F_1, N_0 \prec F_1, S_1 \prec F_1$ with the rest $S_i, N_i \prec 0$ has been imposed.}
\centering
\fbox{%
\begin{tabular}{rrrrrr|rrrr}
 Time &\multicolumn{5}{c}{ \em 1} &\multicolumn{4}{c}{ \em 	2}  \\
  &Gene & \em ZNF &\em CCN &\em SIV &\em SCY & \em ZNF &\em CCN &\em SIV &\em SCY \\
  \hline
&\em ZNF &1.11 & 0.01 &0.01 &0.01 &  -0.42 & 0 & 0&0 \\
&\em CCN & 0.01 & 1.11 & 0.01 & 0.01 &  0 & -0.42 & 0 & 0 \\
  &\em SIV  & 0.01 & 0.01 & 1.11 &  0.01 & 0 & 0 & -0.42 & 0 \\
 & \em SCY  & 0.01 & 0.01 & 0.01 & 1.11 & 0 & 0 & 0 & -0.42\\
   \hline
  &\em ZNF & -0.42 & 0 & 0 & 0 & 1.11 &  0.01 & 0.01 & 0.01 \\
  &\em CCN & 0 & -0.42 & 0 & 0 & 0.01 & 1.11 &  0.01 &  0.01 \\
  &\em SIV  & 0 & 0 & -0.42 & 0 & 0.01 & 0.01 & 1.11 &  0.01 \\
  &\em SCY & 0 & 0 & 0 & -0.42 & 0.01 & 0.01 & 0.01 & 1.11 \\
\end{tabular}}
\end{table}

\subsection{Structured Gaussian graphical models for conditional correlation matrix}

It is generally important that all variables are on comparable scales if a structured model on the concentration matrix is considered so that conclusions are interpretable. In contrast, model based on the conditional correlation matrix have proprieties of invariance under rescaling as showed in Lemma \ref{Invariance}.
\begin{lemma}
\emph{(Invariance)}
\label{Invariance}
The inverse variance matrix $\boldsymbol \Theta$ is not invariant under rescaling $\mathbf{Y}$ by $\mathbf{D}$. Let $\mathbf{D}$ be a diagonal matrix with diagonal entries the scaled precision matrix $\boldsymbol \Omega$ is invariant.
\end{lemma}

\begin{proof}
The inverse variance matrix $\boldsymbol \Theta^*$ of $\mathbf{Y}^*= \mathbf{D} \mathbf{Y}$ is given by
$$\boldsymbol \Theta^* = (\mathbf{D} \boldsymbol \Sigma \mathbf{D})^{-1} = \mathbf{D}^{-1} \boldsymbol \Sigma^{-1} \mathbf{D}^{-1} = \mathbf{D}^{-1} \boldsymbol \Theta \mathbf{D}^{-1},$$
so $$\boldsymbol \Theta = \mbox{var}(\mathbf{Y})^{-1} \neq \mbox{var}(\mathbf{Y}^{*})^{-1} = \boldsymbol \Theta^*,$$
so $\boldsymbol \Theta$ is not invariant under rescaling of $\mathbf{Y}$. The conditional correlation matrix $\boldsymbol \Omega$ is invariant under rescaling, i.e $ \boldsymbol \Omega = \boldsymbol \Omega^*,$ since
\[
\boldsymbol \Omega^* = \boldsymbol \Sigma^{*-\frac{1}{2}}_0 \boldsymbol \Theta^* \boldsymbol \Sigma^{*-\frac{1}{2}}_0 = \boldsymbol \Sigma^{*-\frac{1}{2}}_0 \mathbf{D}^{-1} \boldsymbol \Theta \mathbf{D}^{-1}  \boldsymbol \Sigma^{*-\frac{1}{2}}_0 \boldsymbol \Sigma_0 = \boldsymbol \Sigma^{-\frac{1}{2}}_0 \boldsymbol \Theta \boldsymbol \Sigma^{-\frac{1}{2}}_0 = \boldsymbol \Omega,\]
and $\boldsymbol \Sigma_0^* = \boldsymbol \Sigma_0 \mathbf{D}^{-1}$.
\end{proof}
A factorial graphical model with structured conditional correlation matrix is obtained by restricting elements of $\boldsymbol \Omega$ such that:
\begin{itemize}
\item all diagonal elements of $\boldsymbol \Omega$ (inverse partial variances) must be identical, and
\item all partial correlations corresponding to edges in the same colour class must be identical.
\end{itemize}

Let $\boldsymbol \Theta = \boldsymbol \Sigma_0 \boldsymbol \Omega \boldsymbol \Sigma_0$ be the concentration matrix, where $\boldsymbol \Sigma_0 = \sum_{m=1}^{s_0} \mathbf{X}^{S^m_0} \theta^m$, and
\[
\boldsymbol \Omega = \mathbf{I} + \sum_{i=1}^{n_\Gamma-1}  \sum_{m =1}^{s_i}  \mathbf{X}^{S^m_i} \theta_m + \sum_{i =0}^{n_\Gamma-1} \sum_{l=m+1}^{n_c} \mathbf{X}^{N^l_i} \theta_l.
\]
A structured graphical model for the conditional correlation matrix is obtained by partitioning $\boldsymbol \Omega$, i.e.:
\[
\boldsymbol \Omega = \sum_{m =1}^{n_C} \mathbf{X}^{(m)} \omega_m.
\]
Elements of the conditional correlation matrix $\boldsymbol \Omega$ are related to elements of the natural partitions $S_i, N_i$. Note that edges with the same colours correspond to specific elements of the conditional correlation matrix which are constrained to be equal. Restrictions are non linear in the conditional correlation matrix and an iterative algorithm is considered in subsection \ref{sec:Penalized likelihood for coloured graphs} to estimate $\boldsymbol \Theta$ such that specific elements in the conditional correlation matrix $\boldsymbol \Omega$ are constrained to be equal.

\paragraph{Example 2: Educational study.}Consider the following factorial graphical model for educational study dataset (see Section \ref{sec:motivating-examples} for the description):
\[
[S_0 \prec F_{\Gamma T}, N_0 \prec F_{\Gamma}, S_1 \prec F_1, N_1 \prec 0].
\]
Table \ref{estconcentration} shows the estimated conditional concentration elements (upper triangular) and conditional covariance elements (lower triangular and diagonal). Note that conditional correlation elements are constraint to be equal while concentration elements are different.
\begin{table}
\caption{\label{concentration}Estimated conditional correlations (upper triangular) and conditional covariance (lower triangular and diagonal). The number of parameter estimated is 8 for the diagonal elements (natural partitions $S_0$), 6 for the natural partition $N_0$ and 4 for the natural partition $S_1$. The total number of estimated parameter is 16, that is the number of colours in the graph is 16. Note that the following model $S_0 \prec F_{\Gamma T}, N_0 \prec F_\Gamma, S_1 \prec F_1$ with the rest $S_i, N_i \prec 0$ has been imposed.}
\centering
\fbox{%
\begin{tabular}{rrrrrr|rrrr}
 Time &\multicolumn{5}{c}{ \em 1} &\multicolumn{4}{c}{ \em 	2}  \\
& Subject& \em Contr & \em Auton & \em Infl & \em Prox & \em Contr & \em Auton & \em Infl &\em Prox \\
  \hline
 \multirow{4}{*}{\em 1} & \em Contr & 1.32 & \em 0.13 &\em 0.21 &\em -0.14 &\em 0.34 &\em 0 &\em 0 &\em 0 \\
& \em  Auton & -0.15 & 1.08  &\em -0.07 &\em -0.07 &\em 0 &\em 0.34 &\em 0 &\em 0 \\
& \em Infl &-0.26 & 0.08 & 1.18 & \em -0.07 &\em 0 &\em 0 &\em 0.34 &\em 0 \\
& \em  Prox &0.17 & 0.08 & 0.08 & 1.13  & \em 0 &\em 0 &\em 0 &\em 0.34 \\
\hline
\multirow{4}{*}{\em 2} & \em Contr & -0.43 & 0 & 0 & 0 & 1.22 &\em 0.13 &\em 0.21 &\em -0.14 \\
& \em   Auton &0 & -0.37 & 0 & 0 & -0.15 & 1.12 &\em -0.07 &\em -0.07 \\
& \em Infl & 0 & 0 & -0.40 & 0 & -0.25 & 0.08 & 1.18&\em -0.07 \\
& \em Prox &  0 & 0 & 0 & -0.38 & 0.16 & 0.08 & 0.08 & 1.14 \\
\end{tabular}}
\end{table}

\section{Penalized maximum likelihood}
\label{sec:pmle}

We have reached some important results by considering constraints on the concentration (or conditional correlation) matrix. The number of parameters to be estimated can be considerable reduced. Each sub-network can be interpreted as its corresponding natural partition. However, dynamic genetic graphs are usually sparse which means that few vertices will be connected. We have given the concepts of sparse and dense graphs in Section \ref{sec:fgl} when the number of nodes tends to infinity. A measure of the density of a graph in the finite case can be defined as follow. For undirected graphs, density is:
$$\varrho= \frac{2|\Gamma T|}{|\Gamma T|\,(|\Gamma T| - 1)}.$$
The maximum number of edges is $\frac{1}{2} |\Gamma T| (|\Gamma T|-1)$, so the maximal density is 1 (for complete graphs) and the minimal density is 0.

We could think to estimate a complete graph and to produce multiple hypothesis testing on the edges of the graph. However, model selection and parameter estimation would be done separately in this case and it would bring at instability of the model \cite{breiman1996heuristics}. Alternatively, we consider $\ell_1$-norm penalty on the concentration (or conditional correlation) matrix to induce sparsity. The choice of  a $\ell_1$-norm can be considered the unique one since other $\ell_p$-norm, where $p$ is typically in the range $[0,2]$, are not suitable in high-dimensional data analysis. The estimates being exactly zero for $p\leq 1$ only, while the optimization problem is convex for $p\geq 1$. Hence $\ell_1$-norm occupies a unique position, as $k = 1$ is the only value of $k$ for which variable selection takes place while the optimization problem is still convex and hence feasible for high-dimensional problems \cite{banerjee2008model}.

\subsection{Penalized likelihood for coloured graphs}
\label{sec:Penalized likelihood for coloured graphs}

Consider a coloured graph that specifies a partition of $\boldsymbol \Theta$, then a set of design matrices $\mathbf{X}$ is used to produce linear operators which are necessary to impose linear restrictions on $\boldsymbol \Theta$.  Every design matrix $\mathbf{X}^{(m)}$, $m = 1, \ldots, n_c$ consists of zeroes and ones and induce $p-1$ linear constraints on $\boldsymbol \Theta$ when the number of 1 in $\mathbf{X}^{(m)}$ is $p$. We define a linear map $\mathbf{A} = (\mathbf{A}_1, \ldots, \mathbf{A}_{n_p})$ where each $\mathbf{X}^{(m)}$ induces a set of matrices $\mathbf{A}_1, \ldots, \mathbf{A}_{n_m}$ and an element in $\mathbf{A}^{(i)}$, $i=1, \ldots, n_p$ assumes value $-1$, $0$ or $1$ as described below:

$$\mathbf{a}^{(i)}_{jt,gs} =  \left\{ \begin{array}{rl}
  1 & \mbox{if} \ \ \sum_{jg} \sum_{ts} \mathbf{X}^{(i)}_{jt, gs} \mbox{ before } jt, gs = p-1 \\
 -1 & \mbox{if} \ \ \sum_{jg} \sum_{ts} \mathbf{X}^{(i)}_{jt, gs} \mbox{ before }  jt, gs = p \\
  0 & \mbox{otherwise},
\end{array} \right.
$$
where `before' is meant with respect to the total row-major ordering of the matrices. Now let $n_p$ be the total number of linear constraints, then $\mathbf{A} = \{\mathbf{A}_1, \ldots, \mathbf{A}_{n_p}\}$ and the linear map is expressed as $\mathbf{A}: \mathbb{R}_{\mbox{sym}}^{\Gamma T} \rightarrow \mathbb{R}^{n_p}$ with:
\[
\mathbf{A}(\boldsymbol \Theta) = [\ab{\mathbf{A}_1}{\boldsymbol \Theta}, \ldots, \ab{\mathbf{A}_{n_p}}{\boldsymbol \Theta}]
\]
where $\ab{}{}$ is the Frobenius inner product.

When we derive the objective function we want to take into account the sparsity assumption, i.e. the sum of the absolute values of the precision matrix  should be less than $\rho$. The $\ell_1$-norm constraint $||\boldsymbol \Theta|| \leq \rho$ can be written as a set of linear equality constraints by introducing slack variables $\mathbf{x}^+, \mathbf{x}^- \in \mathbb{R}^k$, i.e.
\begin{eqnarray}
\label{eq:logslack}
(\hat{\boldsymbol \Theta}) &:=& \mathop{\mbox{argmin}}_{\boldsymbol \Theta} \{-l(\boldsymbol \Theta) + \lambda \mathbf{x}^+ + \lambda \mathbf{x}^- \}\\
\nonumber
&&\\
\nonumber
\mbox{subject to}  &&\mathbf{B}(\boldsymbol \Theta) - \mathbf{x}^+ + \mathbf{x}^- = \mathbf{0}\\
\nonumber
&&\boldsymbol \Theta \succeq 0, \mathbf{x}^+, \mathbf{x}^- \geq 0.
\end{eqnarray}
where $\mathbf{B}_i$, for $i = 1, \ldots k$, are symmetric matrices with just element $b_{g,t} = b_{t,g} =1$ and the rest equal to zero, and $k = |\Gamma T| (|\Gamma T| + 1)/2$ or $k = |\Gamma T| (|\Gamma T| - 1)/2$ if diagonal elements are penalized or not, respectively.

Now we want to include the set of linear constraints $\mathbf{A}(\boldsymbol \Theta)$ which are derived by imposing a specific coloured graph. To achieve sparse graph structures with a structured a priori graph one can minimize a convex log-likelihood function, i.e.:
\begin{eqnarray}
\label{eq:loglikpen34}
(\hat{\boldsymbol \Theta}) &:=& \mathop{\mbox{argmin}}_{\boldsymbol \Theta} \{-l(\boldsymbol \Theta) + \lambda \mathbf{x}^+ + \lambda \mathbf{x}^- \}\\
\nonumber
&&\\
\nonumber
\mbox{subject to}  &&A(\boldsymbol \Theta) = \mathbf{0}\\
\nonumber
 &&\mathbf{B}(\boldsymbol \Theta) - \mathbf{x}^+ + \mathbf{x}^- = \mathbf{0}\\
\nonumber
&&\boldsymbol \Theta \succeq 0, \mathbf{x}^+, \mathbf{x}^- \geq 0.
\end{eqnarray}
We have re-written the convex optimization problem in the standard form. This is a quadratic semi-definite log-determinant programming problem which allows to impose factorial graphical models.

The non linearity of the objective function and the positive definiteness of the constraint make the optimization problem (\ref{eq:loglikpen34}) not trivial. We use an algorithm called LogdetPPA to find a solution of (\ref{eq:loglikpen34}). LogdetPPA employs the essential ideas of the proximal point algorithm (PPA), the Newton method and the preconditioned conjugate gradient solver \citep{wang2009solving}.

Note that LogdetPPA algorithm was developed into Matlab and it gives the opportunity to solve convex optimization problems with linear constraints which need to be implemented. We implemented linear constraints for factorial graphical models as described in Section \ref{sec:fgl}. Moreover, it is possible to use function of Matlab within R. In fact, the package $\tt{R.Matlab}$ allows to connect Matlab and R. R is more suitable for statistical analysis and it is an open source software so source codes of packages are available. We took advantage from $\tt{R.Matlab}$ to create a virtual connection between R and Matlab so that we are able to solve the constraint optimization problem within R.
\paragraph{LogdetPPA optimization of FGL concentration model.} For FGL$_\Theta$ problems, we can apply the algorithm logdetPPA after having create the linear operators $\mathbf{A}$ and $\mathbf{B}$. Then, $\boldsymbol \Theta$ is the matrix that minimizes the penalized log-likelihood (\ref{eq:loglikpen34}) among the space of all symmetric $\Gamma T \times \Gamma T$ matrices for whom the non linear restriction on $\boldsymbol \Theta$ holds. For example we used FGL$_\Theta$ to estimate the matrix represented in Table \ref{estconcentration}.
\paragraph{LogdetPPA optimization of FGL conditional correlation model.} This algorithm allows to introduce linear constraint but in case we are modelling conditional correlations the constraints are not linear. However, we overcome this problem by using an iterative algorithm which make use of logdetPPA after having found an initial guess for $\mbox{diag}(\boldsymbol \Theta)$. The pseudo-code is described in Algorithm \ref{ital}.
\begin{algorithm}
\caption{Calculate sparse $\boldsymbol \Theta$ with structured on $\boldsymbol \Omega$}
\label{ital}
\begin{algorithmic}
\REQUIRE Coloured graph $S_i \prec *, N_i \prec *$ and set an initial vector $\boldsymbol \Sigma_0$.
\begin{enumerate}
\item Find the linear maps $\mathbf{A}_1, \ldots, \mathbf{A}_m$.
\item $k = 0, \ldots, $.
\item Estimate $\boldsymbol \Theta^{(k)}$.
\item Set $\boldsymbol \Sigma_0 = \mbox{diag}(\boldsymbol \Theta^{(k)})$.
\item Replay $\boldsymbol \Sigma_0^{(k)}$ with $\boldsymbol \Sigma_0^{(k+1)}$ and estimate $\hat{\boldsymbol \Theta}^{(k+1)}$.
\item If {$|| \hat{\boldsymbol \Theta}^{(k+1)} - \hat{\boldsymbol \Theta}^{(k)}||_1 < \epsilon$} Stop \\
	  else $k = k+1$, go to 3 \\
      end.
\end{enumerate}
\end{algorithmic}
\end{algorithm}
Usually, the convergence is reached after few iterations (4-10). By taking a starting point $\boldsymbol \Sigma_0$ the optimization problem (\ref{eq:loglikpen34}) is a convex optimization problem, i.e., the objective function is convex on $\boldsymbol \Theta$, and the feasible region is convex.

\section{Model selection}
\label{sec:modelselection}
We have seen that estimation of $\boldsymbol \Theta$ considering different coloured graphs given a smoothing parameter $\lambda$ is possible inside a convex optimization framework. We have also proposed several factorial graphical models. In this subsection we address same issues on how to choose the 'best' coloured graph, and what should be a good compromise between a sparse and a dense graph. In particular, according to \cite{meinshausen2010stability} we want to find a smoothing parameter such that the expected number of false positive links is taken under control. This aim can be reached through stability selection. Stability selection is similar to bootstrap idea which consists of a re-sampling procedure.  For example we used FGL$_\Omega$ to estimate the matrix represented in Table \ref{concentration}.

\subsection{Classical approaches}
Let's assume that the smoothing parameter is fixed (point-wise control) $\lambda = \lambda_{opt}$ and two coloured graphs need to be compared:
\[
[S_0 \prec F_{\Gamma T}, N_0 \prec F_{\Gamma}, S_1 \prec F_1, N_1 \prec 0],
\]
and
\[
[S_0 \prec F_{\Gamma T}, N_0 \prec F_{\Gamma}, S_1 \prec F_{\Gamma}, N_1 \prec 0].
\]
An information criterion such as AIC, BIC and AICc can be used to compare different factorial graphical models. AIC typically will select more and more complex models as the sample size increases, because the maximum log-likelihood increases linearly with $n$ while the penalty term for complexity is proportional to the number of degrees of freedom.  Note that for factorial graphical models the number of degrees of freedom is approximated by the number of "free" parameters different from zero, i.e.  the estimated elements different from zero which do not belong to the same partition. AICc penalizes complexity more strongly than AIC, with less chance of over-fitting the model. BIC is constructed in a manner quite similar to AIC with stronger penalty for complexity \citep{claeskens2008model}.

\subsection{Stability selection}

Usually, the choice of smoothing parameter $\lambda$ is crucial as it leads to the structure of the network. In particular, if $\lambda = 0$ and there are no constraints, the maximum likelihood estimator is $\hat{\boldsymbol \Theta} = \mathbf{S}^{-1}$, if $\lambda \rightarrow \infty$, $\hat{\boldsymbol \Theta}$ is diagonal which means all random variables are independent.
Generalized cross validation can be used to select the tuning parameter $\lambda = \lambda_{cv}$ but \cite{leng2006note} showed that given $\lambda_{cv}$ the estimator of $\boldsymbol \Theta$ is not consistent in terms of variables selection. Alternatively, bootstrap \citep{breiman1999prediction} can be considered to estimate empirical distribution of each element of $\hat{\boldsymbol \Theta}$. We prefer stability selection \citep{meinshausen2010stability} because the expected number of links falsely estimated is controlled and variables selection is consistent. Moreover, the choice of a smoothing parameter $\lambda$ becomes less important \cite{meinshausen2010stability}. We adapt stability selection to factorial graphical models.

Let's consider a vector of smoothing parameters such that $\lambda \in \Lambda \subseteq \mathbb{R^+}$ that determines the amount of regularization.
\begin{theorem}
\label{propfast}
A necessary and sufficient condition for $\theta_{ij,kl} = 0$ for all $i,k \in \Gamma$ and $j,l \in T$ is that $|S_{ij, kl}| \leq \lambda$ for all $i \neq k$ and $l \neq j$ \cite{mazumder2011exact}.
\end{theorem}
Upper and lower bounds of $\lambda$, $u_\lambda$ and $l_\lambda$ respectively, are calculated such that $u_\lambda = max(|S_{ij, kl}|)$  and $l_\lambda = min(|S_{ij,kl}|)$ for $i \neq k$ $l \neq j$, where $i, k\in \Gamma$ and $j,l \in T$. Then, for all $\lambda > u_\lambda$ an empty graph is estimated while for $\lambda < l_\lambda$ a fully connected graph is estimated. We search the solution of the optimization problem (\ref{eq:loglikpen34}) for value of $\lambda$ into the range $[l_\lambda, u_\lambda]$. The ``optimal'' value of $\lambda = \lambda_{opt}$ can be chosen by minimizing a score which measures the goodness-of-fit.

Let graph $\hat{G} = (V,  \hat{E})$ be inferred, where
$$\hat{E}_{\lambda_{opt}} = \{(v_i,v_j): \hat{\theta}_{i,j} \neq 0 \}$$
is the estimated edge set and
$$E = \{(v_i,v_j): \theta_{i,j} \neq 0\}$$ denotes the active set. Let $I = \{1, \ldots, n\}$ be the index set for sample $\mathbf{y}^{(i)}$, $i \in I$, then:
\begin{algorithm}
\caption{Stability selection for graphical models.}
\label{italb}
\begin{algorithmic}
\REQUIRE n.
\begin{itemize}
\item Draw sub-samples of size [n/2] without replacement, denoted by $I^* \subset \{1, \ldots, n\}$, where $|I^*| = [n/2]$.
\item Run the selection algorithm $\hat{E}_{\lambda_{opt}}(I^*)$ on $I^*$.
\item Do these steps many times and compute the relative frequencies, $$\hat{\Pi}^{\lambda_{opt}}_{ij} = P^*((v_i, v_j) \in \hat{E}_{\lambda}), \mbox{for} \ \  i,j = 1, \ldots, p.$$
\end{itemize}
\end{algorithmic}
\end{algorithm}
The set of stable edges is indicated as
$$\hat{E}_{stable} = \{(v_i,v_j): \hat{\Pi}_{ij}^{\lambda} \geq \pi_{thr}\},$$
and it depends on $\lambda_{opt}$ via $\hat{\Pi}_{ij}^{\lambda_{opt}}$. The tuning parameter $\pi_{thr}$ indicates a threshold and controls the expected number of falsely selected links. Assume that the joint distribution of the random variables is exchangeable and $\hat{E}$ is a better choice than a random guessing, then it can be shown that
$$\mathbb{E}(FP) \leq \frac{1}{2 \pi_{thr} - 1} \frac{q^2}{k},$$
where $k$ is the dimension of the model (it depends on the factorial model), $q$ is the number of selected variables (e.g. $|\hat{E}|$), and $FP = |E^c \cap \hat{E}_{stable}|$ is the number of falsely positive selected. This is a finite sample control, even if $k >> N$. Choose $\mathbb{E}(FP) \leq v$, then if $q^2 \leq v k$:
$$\pi_{thr} = (1 + q^2/v k)/2,$$
and $\pi_{thr} \in (\frac{1}{2}, 1)$ is bounded.

 \begin{figure}
    \centering
    \makebox{\includegraphics[scale=0.8]{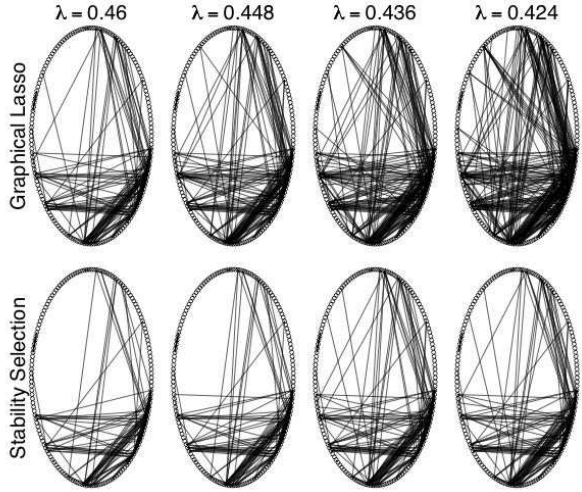}}
\caption{\label{stabgl111}Graph selection with cross validation and stability selection procedure.}
\end{figure}

Figure \ref{stabgl111} is taken from \cite{meinshausen2010stability} and it illustrates that the choice of a tuning parameter $\lambda$ is less important with stability selection than cross validation.

\section{Simulation study}
\label{sec:simulationstudy}

We considered a simulation study to show the performance of the proposed model. Table \ref{tab:simulstudyscheme0} shows the simulation study scheme in which four different scenarios are studied. Here for different scenarios we mean that the number of nodes, links or time points change while the structure of the networks is the same.
\begin{table}
\caption{\label{tab:simulstudyscheme0}Simulation study scheme in which four scenarios are represented. The first column is an identification number, the second one indicates the number of variables per each time point (third column). The number of independent samples are represented in the last column.}
\centering
\fbox{%
\begin{tabular}{*{6}{c}}
ID&g &g*& t&  p & n\\
\hline
1&20& 0 & 3 & 60&50 \\
2&- & 20& - & 120& -\\
3&- & 40& - & 180& -\\
4&- & 60& - & 240& -\\
\end{tabular}}
\end{table}

We are mainly interested in the performance of our estimator in terms of false positive (FP), false negative (FN), true positive (TP) and true negative (TN). Let $\boldsymbol \Theta$ be the true and $\hat{\boldsymbol \Theta}$ be the estimated precision matrix. These measure summarize:
\begin{itemize}
\item the percentage of links that are falsely estimated, i.e. an element in $\boldsymbol \Theta$ is zero but is not zero in $\hat{\boldsymbol \Theta}$ (FP) and an element in $\boldsymbol \Theta$ is not zero but is zero in $\hat{\boldsymbol \Theta}$ (FN),
\item the percentage of links that are positively estimated, i.e. an element in $\boldsymbol \Theta$ is not zero and it is not zero in $\hat{\boldsymbol \Theta}$ (TP) and  an element in $\boldsymbol \Theta$ is zero and it is zero in $\hat{\boldsymbol \Theta}$ (TN).
\end{itemize}
Other measure like these are the false discovery and false not discovered.

Whereas, we are not looking at the performance of our estimator in terms of distances between the "true" parameter $\boldsymbol \Theta$ and the estimated ones $\hat{\boldsymbol \Theta}$. Another element of interest is the $sign$ of the estimated conditional covariance. In fact, this element is not zero, if $-sign(\theta_{ij})$ is positive this indicates a positive dependence whereas if $-sign(\theta_{ij})$ is negative it indicates a negative dependence.

For each scenario we simulate 100 datasets from a multivariate normal distribution with $\boldsymbol \mu$ equal to zero and $\boldsymbol \Sigma$ equal to the inverse of a precision matrices $\boldsymbol \Theta$ where the coloured graph is simulated from the following model:
\[
[S_0 \sim F_{\Gamma T}, N_0 \sim  F_{\Gamma}, S_1 \sim  F_1],
\]
while the rest is zero. Note that we keep the true network constant but we increase the number of nodes in the graph. Random variables associated with these added nodes are independent. We keep the number of replicates and time points constants. The number of replicates is fewer than the number of random variables.
\begin{table}
\caption{\label{tab:perfomcheck0}The average of the proportions of how many links have been correctly estimated were calculated by the False Positive (FP), False Negative (FN), False Discovery (FD) and False not Discovery (FnD).}
\centering
\fbox{%
\begin{tabular}{|rr|rrrr|}
  \hline
   & &  $\overline{FP}$ & $\overline{FN}$ & $\overline{FD}$ & $\overline{FnD}$  \\
  \hline
 &\color{red}{AICc} & 0.0092 & 0.0811 & 0.2000 & 0.0031  \\
 1&BIC & 0.0363 & 0.0139 & 0.4873 & 0.0005 \\
 &AIC & 0.0698 & 0.0069 & 0.6470 & 0.0003  \\
  \hline
 &\color{red}{AICc}  & 0.0057 & 0.0447 & 0.2899 & 0.0006  \\
2 &BIC & 0.0088 & 0.0321 & 0.3826 & 0.0005  \\
 &AIC & 0.0437 & 0.0041 & 0.7514 & 0.0001 \\
   \hline
  &\color{red}{AICc} & 0.0016 & 0.4585 & 0.2730 & 0.0036  \\
 3& BIC  & 0.0016 & 0.4585 & 0.2730 & 0.0036  \\
  &AIC  & 0.0288 & 0.1452 & 0.8088 & 0.0012 \\
\hline
& \color{red}{AICc} & 0.0091 & 0.1034 & 0.1680 & 0.0052 \\
4&  BIC & 0.0396 & 0.0517 & 0.4527 & 0.0027  \\
&AIC  & 0.0670 & 0.0000 & 0.5704 & 0.0000 \\
\hline
\end{tabular}}
\end{table}
Table \ref{tab:perfomcheck0} shows the average over 100 of measures that is $\overline{FP} = \sum_{i=1}^{100} FP_i$, $\overline{FD} = \sum_{i=1}^{100} FD_i$ and so on. Moreover we show that AICc performs better on average and other graphical lasso such us proposed by \cite{tibshirani1996regression} does not perform well in case of structured dynamic graphical models (see Table \ref{tab:averperf0}).

\begin{table}
\caption{\label{tab:averperf0}Average performance for neighbourhood selection model, graphical lasso and structured graphical lasso.}
\centering
\fbox{%
\begin{tabular}{|rr|rrrr|}
  \hline
              &        & $\overline{FP}$ & $\overline{FN}$ & $\overline{FD}$ & $\overline{FnD}$  \\
  \hline
            1 & glasso   &  0.002 & 0.977 & 0.667  & 0.035 \\
              & FGL &  0.006 & 0     & 0.137  & 0     \\

  \hline

            2 & glasso   & 0.001  & 0.964 & 0.620 & 0.013\\
              & FGL & 0.005  & 0     & 0.263 & 0   \\
\hline
            3 & glasso   & 0.001 & 0.988 &  0.906& 0.008 \\
              &FGL  & 0.001 & 0.458 & 0.273 & 0.004 \\
              \hline
\end{tabular}}
\end{table}

\section{FGL data analysis}
\label{sec:dataanalysis}
We have seen that several factorially coloured graphs can be imposed on a graphical model and a model selection procedure is necessary to select the ``best'' dynamic graph. Moreover, a smoothing parameter $\lambda$ that regulates the sparsity needs to be selected.

Let's consider several coloured graphs for Human T-cell dataset. Information criterion measures, for each of the top ten models, are showed in Table \ref{modelsel4}.
\begin{table}
\caption{\label{modelsel4}Model ordering according to AICc for T-cell dataset}
\centering
\fbox{%
\begin{tabular}{*{9}{c}}
 &\multicolumn{3}{c}{ \em Model}    & &AIC     &  AICc   \\
  \hline
  $S_0 \sim F_1 $& $ N_0 \sim F_\Gamma $&    $S_1 \sim F_T$&$ N_1 \sim F_\Gamma   $  &$ S_2 \sim F_T  $  &175.82 & 178.19 \\
  $S_0 \sim F_1  $&$ N_0 \sim 0 $&    $S_1 \sim F_\Gamma$&$  N_1 \sim F_\Gamma  $  & $ S_2 \sim 0 $  & 175.84 & 178.42\\
  $S_0 \sim F_1$&   $N_0 \sim F_{\Gamma T}$&$ S_1 \sim F_\Gamma$&$  N_1\sim F_{\Gamma T}  $  & $ S_2 \sim 0 $  & 176.32 & 178.60  \\
  $S_0 \sim F_1$&   $N_0 \sim F_\Gamma$& $S_1 \sim F_1$&$  N_1\sim 0$  & $ S_2 \sim 0 $  &176.78 & 178.93 \\
  \hline
  $S_0 \sim 1$&   $N_0 \sim F_\Gamma$& $S_1 \sim F_\Gamma$&$ N_1\sim F_\Gamma    $  &$ S_2 \sim 0 $  & 177.04 & 179.19\\
  $S_0 \sim 1$&   $N_0 \sim F_{\Gamma T} $& $S_1 \sim F_1    $&  $  N_1\sim F_{\Gamma T}  $  & $ S_2 \sim 0 $  &176.81 & 179.81 \\
\end{tabular}}
\end{table}
We select the first coloured graph since it has the smallest AICc. This model is described in Figure \ref{modelfortcell}.
\begin{figure}
\centering
\makebox{\includegraphics[scale=0.55]{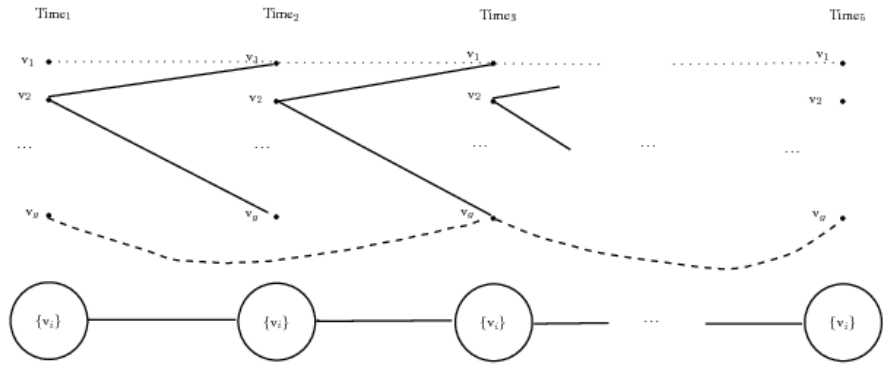}}
\caption{\label{modelfortcell}Selected model after model selection for T-cell dataset.}
\end{figure}
According to the coloured graph this scheme summarize the characteristics of the estimated graph. The networks at temporal lag 0 are constrained to be equal across the five observed time points. Moreover, the networks at temporal lag 1 are constrained to be equal across time. There are no links between time $t$ and $t+2$ since we are assuming conditional independence between time $t$ and $t+2$ except for self interactions, i.e. interactions between the same couple of genes.
Figure \ref{graphN0} shows interactions between genes at lag 0 (left part of the figure), and it shows interactions between genes at lag 1 (right part of the figure).

\begin{figure}[htbp]
\includegraphics[scale=0.5]{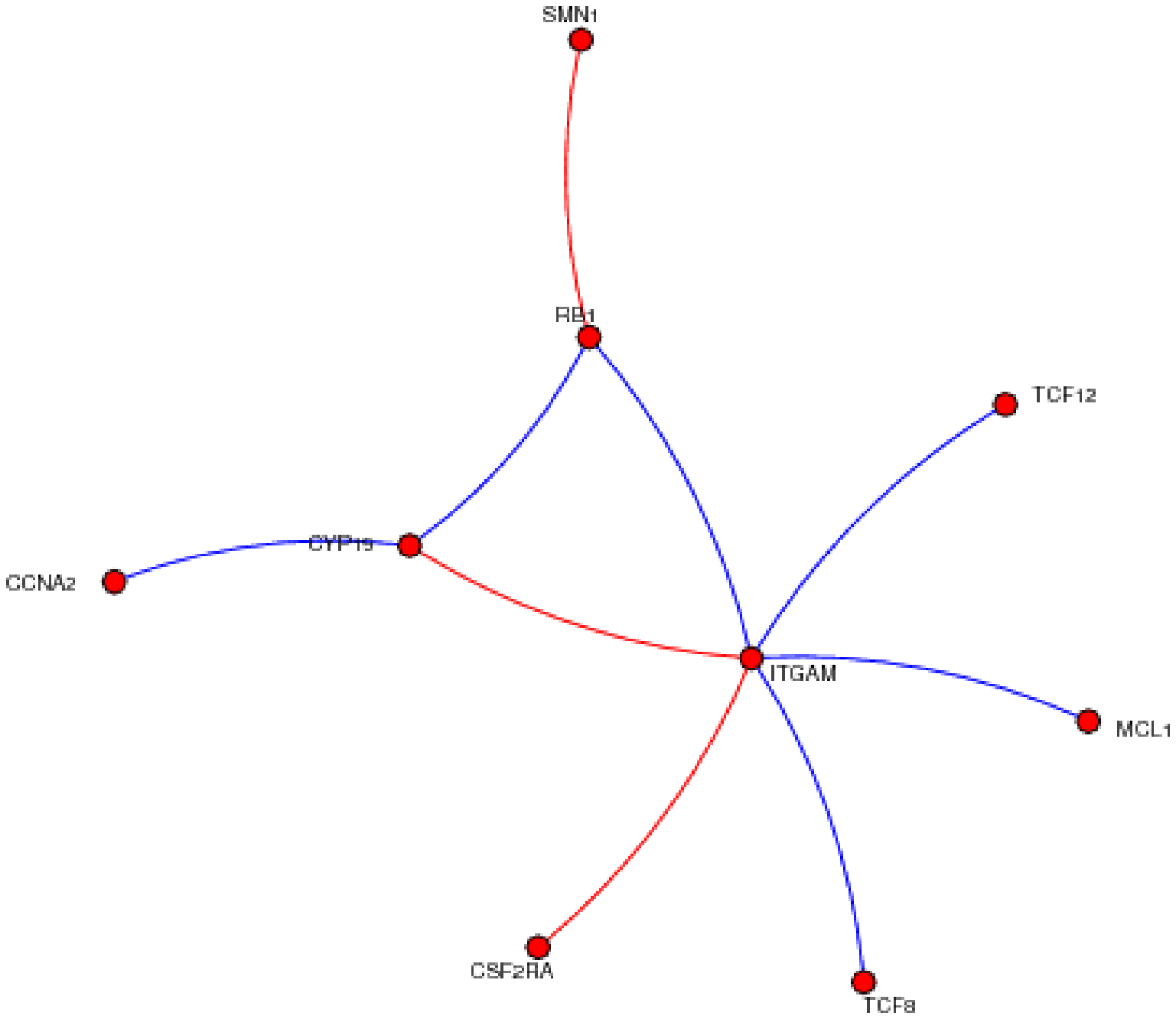}
\includegraphics[scale=0.5]{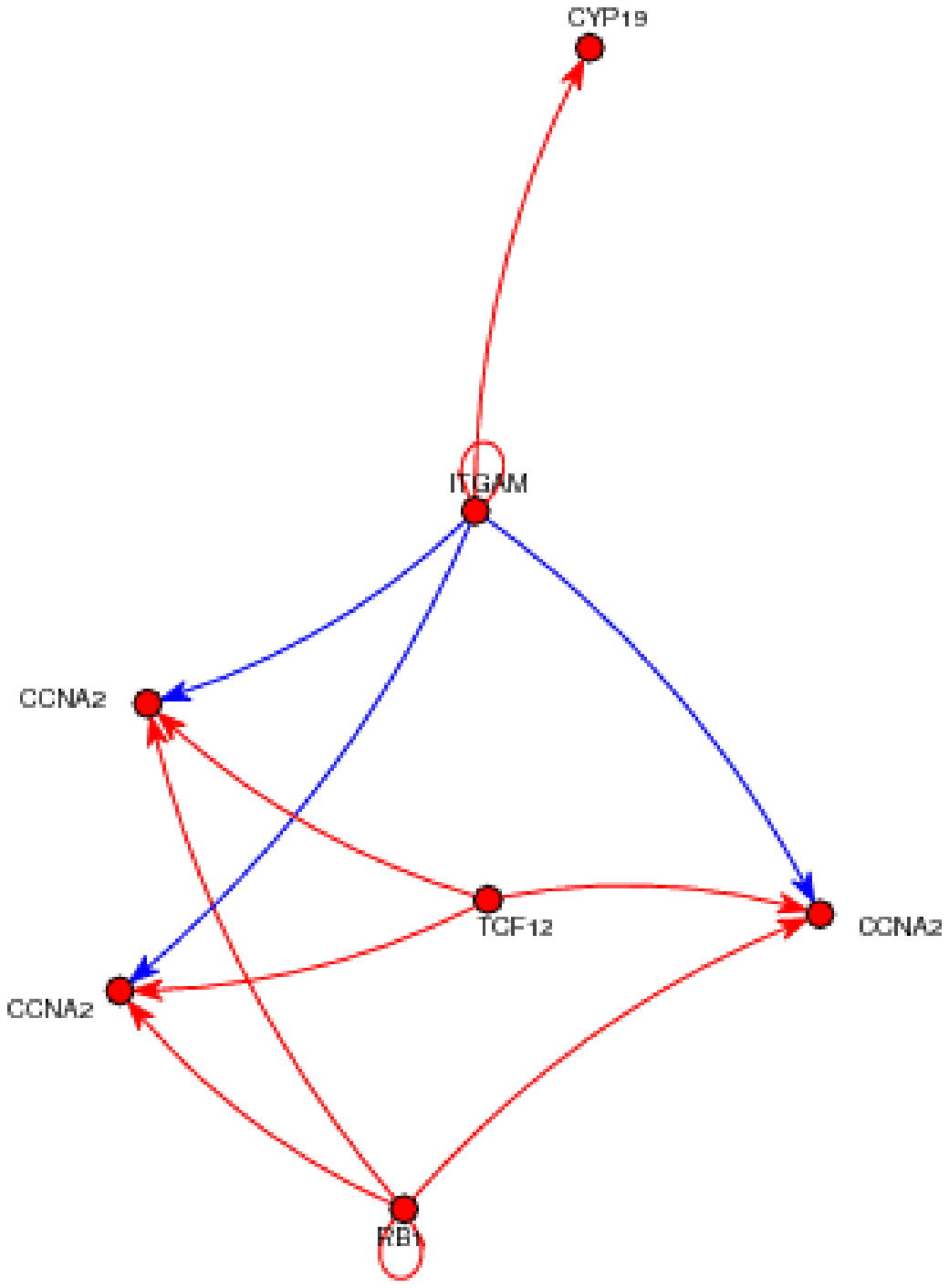}
 \centering
 \caption[Caption]{Representation of interactions between genes at temporal lag 0. Note that networks at lag 0 at time $1,2,3,\ldots, 5$ are equal since we impose $N_0 \sim F_\Gamma$(top). Representation of interaction between genes at temporal lag 1. Note that networks at lag 1 between time $(1,2), (2,3), (3,4), (4,5)$ are equal since we impose $N_1 \sim F_\Gamma$(bottom).}
 \label{graphN0}
 \end{figure}


\section{Discussion}

As more and more large datasets become available, the need for efficient tools to analyze such data has become imperative. In this chapter, we have considered sparse dynamic Gaussian graphical models with $\ell_1$-norm penalty. This type of modelling offers a straightforward interpretation, i.e. the edges of the graph defining the partial conditional correlations among the nodes. In particular, under the sparsity assumption, a large part of the precision matrix can be filled with zeros a priori. Based on the consideration of dynamic and model-oriented definitions, we are able to reduce the number of parameters to be estimated. We have shown that FGL$_\Theta$ proved to be powerful on both simulated and real data analysis.

\end{document}